\documentclass[a4paper,UKenglish]{lipics-v2018}

\usepackage{microtype}
\usepackage[utf8]{inputenc}

\usepackage{xspace}
\usepackage{amsmath}
\usepackage{amsthm}
\usepackage{soul}
\usepackage[table]{xcolor}

\hideLIPIcs

\newcommand\ttrue{\ensuremath{\mathtt{true}}}
\newcommand\ffalse{\ensuremath{\mathtt{false}}}
\newcommand\B{\ensuremath{\mathtt{b}}}
\newcommand\E{\ensuremath{\mathtt{e}}}
\newcommand\C{\ensuremath{\mathtt{c}}}
\newcommand\D{\ensuremath{\mathtt{d}}}
\newcommand\zero{\ensuremath{\mathtt{0}}}
\newcommand\one{\ensuremath{\mathtt{1}}}
\newcommand\notmodels{\,\not\!\models}
\newcommand\pmlg{\textsc{pmlg}\xspace}
\newcommand\sat{\textsc{sat}\xspace}
\newcommand\weth{\textsc{eth}\xspace}
\newcommand\seth{\textsc{seth}\xspace}
\newcommand{\poly}{\operatorname{poly}}

\title{On the Complexity of Exact Pattern Matching in Graphs: Binary Strings and Bounded Degree}

\titlerunning{On the Complexity of Exact Pattern Matching in Graphs}

\author{Massimo Equi}{Department of Computer Science, University of Helsinki, Finland}{massimo.equi@helsinki.fi}{}{}

\author{Roberto Grossi}{Dipartimento di Informatica, Universit\`a di Pisa, Italy}{grossi@di.unipi.it}{}{}

\author{Veli M\"{a}kinen}{Department of Computer Science, University of Helsinki, Finland}{veli.makinen@helsinki.fi}{}{}

\authorrunning{M. Equi, R. Grossi, V. M\"{a}kinen}

\Copyright{M. Equi, R. Grossi, V. M\"{a}kinen}

\subjclass{
~
~
}

\category{}

\relatedversion{}

\supplement{}
\funding{This work has been partially supported by Academy of Finland (grant 309048)}

\EventEditors{XYZ}
\EventNoEds{3}
\EventLongTitle{XYZ}
\EventShortTitle{XYZ}
\EventAcronym{MFCS}
\EventYear{2018}
\EventDate{August 27--31, 2018}
\EventLocation{Liverpool, GB}
\EventLogo{}
\SeriesVolume{117}
\ArticleNo{84} 
\nolinenumbers 

\begin{document}

\maketitle

\begin{abstract}
Exact pattern matching in labeled graphs is the problem of searching paths of a graph $G=(V,E)$ that spell the same string as the pattern $P[1..m]$. This basic problem can be found at the heart of more complex operations on variation graphs in computational biology, of query operations in graph databases, and of analysis operations in heterogeneous networks, where the nodes of some paths must match a sequence of labels or types. We describe a simple conditional lower bound that, for any constant $\epsilon>0$, an $O(|E|^{1 - \epsilon} \, m)$-time or an $O(|E| \, m^{1 - \epsilon})$-time algorithm for exact pattern matching on graphs, with node labels and patterns drawn from a binary alphabet, cannot be achieved unless the Strong Exponential Time Hypothesis (\seth) is false. The result holds even if restricted to undirected graphs of maximum degree three or directed acyclic graphs of maximum sum of indegree and outdegree three. Although a conditional lower bound of this kind can be somehow derived from previous results (Backurs and Indyk, FOCS'16), we give a direct reduction from \seth for dissemination purposes, as the result might interest researchers from several areas, such as computational biology, graph database, and graph mining, as mentioned before.
Indeed, as approximate pattern matching on graphs can be solved in $O(|E|\,m)$ time, exact and approximate matching are thus equally hard (quadratic time) on graphs under the \seth assumption. In comparison, the same problems restricted to strings have linear time vs quadratic time solutions, respectively, where the latter ones have a matching \seth lower bound on computing the edit distance of two strings (Backurs and Indyk, STOC'15).
\end{abstract}

\section{Introduction}
\label{sec:intro}

Large-scale labeled graphs are becoming ubiquitous in several areas, such as computational biology, graph databases, and graph mining. Applications require sophisticated operations on these graphs, and often rely on primitives that locate paths whose nodes have labels or types matching a pattern given at query time.

In graph databases, query languages provide the user with the ability to select paths based on the labels of their nodes or edges, where the edge labels are called properties. In this way, graph databases explicitly lay out the dependencies between the nodes of data, whereas these dependencies are implicit in classical relational databases~\cite{AnglesGutierrez2008}.
Although a standard query language has not been yet universally adopted (as it occurred for SQL in relational databases), popular query languages such as Cypher~\cite{FrancisGGLLMPRS18}, Gremlin~\cite{Rodriguez15}, and SPARQL~\cite{sparqlquery} offer the possibility of specifying paths by matching the labels of their nodes.

In graph mining and machine learning for network analysis, heterogeneous networks specify the type of each node~\cite{ShiLZSY17}. For example, in the DBLP network~\cite{YangL12}, the nodes for authors can be marked with letter 'A', and the nodes for papers can be marked with letter 'P', where edges connect authors to their papers. For example, coauthors can be identified by the pattern 'APA' when it matches two different nodes with 'A'. The strings generated by the labels on the paths have several applications in heterogeneous networks, such as graph kernels~\cite{HidoK09} or node similarity~\cite{ConteFGMSU18}, where  a basic tool is retrieving the paths for a string.

In genome research, the very first step of many standard analysis pipelines of high-throughput sequencing data has been to align the sequenced fragments of DNA (called reads) on a reference genome of a species. Further analysis reveals a set of positions where the sequenced individual differs from the reference genome. After years of these kind of studies, there is now a growing dataset of frequently observed differences between individuals and the reference. A natural representation of this gained knowledge is a \emph{variation graph} where the reference sequence is the backbone and variations are encoded as alternative paths \cite{Sch09}. Aligning reads (pattern matching) on this labeled graph gives the basis for the new paradigm called \emph{computational pan-genomics} \cite{Maretal18}. There are already tools that use such ideas, e.g.~\cite{vg}.

Although there is a growing need to perform pattern matching on graphs in several situations described above, the idea of extending the problem of string searching in sequences to pattern matching in graphs was studied over 25 years ago as a search problem in \emph{hypertext} \cite{manber1992approximate}. The history of key contributions is given in Table~\ref{table:summary}, where the two best known results for exact and approximate pattern matching, both taking quadratic time in the worst case, are highlighted. Note that errors in the graphs makes the problem NP-hard~\cite{AmirLL00}, so we consider errors in the pattern only.

\definecolor{mylightgray}{gray}{0.93}

\begin{table}[t]
\centering
\begin{tabular}{|c|c|c|c|c|}
     \hline
     \multicolumn{5}{|c|}{State of the art for \emph{PMLG}}\\
     \hline\hline
     \textbf{Year} & \textbf{Authors} & \textbf{Graph} & \textbf{Exact/} & \textbf{Time}\\
     ~ & ~ & ~ & \textbf{Approximate} & ~\\
     \hline
     1992 & Manber, Wu \cite{manber1992approximate} & DAG & approximate$^\textit{(1)}$ & $O(m|E| + occ\lg\lg m)$\\
     \hline
     1993 & Akutsu \cite{akutsu1993linear} & Tree & exact\phantom{$^\textit{(3)}$} & $O(N)$\\
    \hline
    1995 & Park, Kim \cite{park1995string} & DAG & exact$^\textit{(3)}$ & $O(N + m|E|)$\\
    \hline
    \rowcolor{mylightgray}
    1997 & Amir et al. \cite{AmirLL00} & general & exact\phantom{$^\textit{(3)}$} & $O(N + m|E|)$\\
    \hline
    1997 & Amir et al. \cite{AmirLL00} & general & approximate$^\textit{(2)}$ & NP-Hard\\
    \hline
    1997 & Amir et al. \cite{AmirLL00} & general & approximate$^\textit{(1)}$ & $O(Nm\lg N + m|E|)$\\
    \hline
    1998 & Navarro \cite{navarro2000improved} & general & approximate$^\textit{(1)}$ & $O(Nm + m|E|)$\\
    \hline
    2017 & Vadaddi et al.~\cite{Vaddadi2017sequence} & general & approximate$^\textit{(1)}$ & $O((|V|+1)m|E|)$\\
    \hline
    \rowcolor{mylightgray}
    2017 & Rautiainen, Marschall \cite{rautiainen2017aligning} & general & approximate$^\textit{(1)}$ & $O(N + m|E|)$\\
    \hline
    2019 & Jain et al. \cite{JZGA19} & general & approximate$^\textit{(2)}$ & NP-Hard on binary alphabet\\
    \hline
\end{tabular}
\caption{Legend: $V$ = set of nodes, $E$ = set of edges, $occ$ =  number of matches for the pattern in the graph, $m$ = pattern   length, $N$ = total length of text in all nodes, \textit{(1)} errors only in the pattern, \textit{(2)} errors in the graph, \textit{(3)} matches span only one edge. The two rows highlighted in gray report the best known bounds for exact and approximate pattern matching.}
\label{table:summary}
\end{table}

A common feature of the bounds reported in Table~\ref{table:summary} is the appearance of the quadratic term $m \, |E|$ (except for the special cases of trees and the general NP-hard approximate version). Here $m$ is the length of the pattern string and $E$ is the set of edges of the graph. 
The quadratic cost of the approximate matching on graphs by Rautiainen and Marschall~\cite{rautiainen2017aligning} are asymptotical optimal under the Strong Exponential Time Hypothesis~\cite{IP01} (\seth) as (i)~they solve the approximate string matching as a special case, since a graph consisting of just one path of $|E|+1$ nodes and $|E|$ edges is a text string of length $n=|E|+1$, and (ii)~it has been recently proved that the edit distance of two strings of length $n$ cannot be computed in $O(n^{2-\epsilon})$ time, for any constant $\epsilon>0$, unless \seth is false~\cite{BI15}. Hence this conditional lower bound explains why the $O(m|E|)$ barrier has been difficult to cross. 

We can only explain the complexity on \emph{approximate} pattern matching on graphs, but nothing is known on \emph{exact} pattern matching on graphs. Indeed, the classical exact pattern matching with a pattern and a text string can be solved in linear time \cite{KMP77}, so one could expect the corresponding problem on graphs to be easier than approximate pattern matching. 

In this paper we end up with a slightly surprising observation that \emph{exact and approximate pattern matching are equally hard on graphs}. Namely, we show the conditional lower bound that an $O(|E|^{1 - \epsilon} \, m)$-time or an $O(|E| \, m^{1 - \epsilon})$-time algorithm for exact pattern matching on graphs cannot be achieved unless \seth is false. This result explains why it has been difficult to find indexing schemes for graphs with other than best case or average case guarantees for fast exact pattern matching \cite{SVM14,GMS17}.

Before going on to give the overview and details of the reduction, let us now fix the problem definition and \seth formulation.

\begin{definition}\label{definition:labeledgraph}
    Given an alphabet $\Sigma$, a labeled graph $G$ is a triplet $(V,E,L)$ where $(V,E)$ is a directed or undirected graph and $L: V \rightarrow \Sigma^+$ is a function that defines which string over $\Sigma$ is assigned to each node.
\end{definition}

\begin{definition}
    Let $u_1, \ldots, u_j$ be a path in graph $G$ and $P$ be a pattern. Also, $L(u)[l:]$ and $L(u)[:l']$ denote the suffix of $L(u)$ starting at position $l$ and the prefix of $L(u)$ ending at position $l'$, respectively. We say that $u_1, \ldots, u_j$ is a match for $P$ in $G$ with offset $l$ if the concatenation of the strings $L(u_1)[l:] \cdot L(u_2) \cdot \ldots \cdot L(u_{j-1}) \cdot L(u_j)[:l']$ equals $P$, for some $l'$.
\end{definition}

The  \emph{Pattern Matching in Labeled Graphs} (\emph{\pmlg}) problem is then defined as:
\begin{description}
\item{\textsc{input}:} a labeled graph $G = (V,E,L)$ and a pattern $P$ over an alphabet $\Sigma$.
\item{\textsc{output}:} all the matches for $P$ in $G$.
\end{description}

For example, in Fig.~\ref{figure:G_Fgadget} pattern \C\D\E\ has two occurrences but pattern \B\C\C\C\E\ does not occur. For our purpose it would be enough to exploit a decision version of the problem, namely, to be able to determine whether or not there exists at least one match for $P$ in $G$, without reporting all of them. Note that the matching path for $P$ can go through the same nodes multiples times in $G$ as otherwise \pmlg is trying to solve the NP-hard Hamiltonian path problem.

We now recall what is \seth, namely, the Strong Exponential Time Hypothesis \cite{IP01}. This is a conjecture which is commonly used as a basis of reductions in the scientific community, even though its weaker version \weth is more widely accepted.
\begin{definition}[\cite{IP01}]
    Let \emph{q-\sat} be an instance of \sat  with at most \emph{q} literals per clause. 
    Given $\delta_q = \inf \, \{\alpha \, : \, \text{there is an} \, O(2^{\alpha n})\text{-time algorithm for\ } q\text{-\sat} \}$, 
    \seth claims that $\lim\limits_{q \to \infty} \delta_q = 1$.
\end{definition}

In other words, it is hard to find an $O(2^{\alpha n})$-time algorithm for general \sat for a constant $\alpha < 1$. We use \seth in the the following result,
given that the best known algorithm for \pmlg, devised 20 years ago~\cite{AmirLL00}, has an $O(|E| \, m)$ time complexity.
\begin{theorem}
\label{theorem:Emlowerbound} 
For any constant $\epsilon > 0$, the Pattern Matching in Labeled Graphs (\pmlg) problem for an alphabet of at least $4$ symbols cannot be solved in either $O(|E|^{1-\epsilon} \, m)$ or $O(|E| \, m^{1-\epsilon})$ time unless \seth is false.
\end{theorem}

We can further strengthen the statement of this theorem by proving the following corollaries.

\begin{corollary}
\label{cor:general-undirected}
The conditional lower bound stated in Theorem \ref{theorem:Emlowerbound} holds even if it is restricted to graphs with binary alphabet for the labels, where each node has degree at most three.
\end{corollary}

\begin{corollary}
\label{cor:general-dag}
The conditional lower bound stated in Theorem \ref{theorem:Emlowerbound} holds even if it is restricted to labeled directed acyclic graphs (DAGs) with binary alphabet for the labels, where each node has the sum of indegree and outdegree at most three.
\end{corollary}

In order to achieve these results we break down our reasoning process in
some intermediate steps. Since this is a conditional lower bound we will reduce \sat to \pmlg. Then we will show that having a truly subquadratic algorithm for \pmlg would cause to solve \sat in $O(2^{\alpha n})$ time with $\alpha < 1$. Our reduction costs $\Tilde{O}(2^{\frac n2})$ time for a \sat formula with $n$ variables and $k = O(\poly(n))$ clauses, where $\Tilde{O}$ is the shorthand for ignoring polynomial factors in $n$, e.g. $O(k n^2 \, 2^\frac n2) = \Tilde{O}(2^\frac n2)$. Hence the main steps can be synthesized as follows.
\begin{itemize}
\item Find a reduction from \sat to \pmlg.
\item Ensure that this reduction costs $\Tilde{O}(2^{\frac n2})$ time.
\item Show that having a $O(|E|^{1 - \epsilon} \, m)$ or a $O(|E| \, m^{1 - \epsilon})$ time algorithm for \pmlg gives a solution for \sat that makes \seth fail.
\end{itemize}

Our reduction shares some similarities with those for string problems in \cite{BI15,BK15,ABW15,BI16,BZ17} as it uses \seth. The closest connection is with \cite{BI16}, where regular expression matching is studied (graph $G_F$ in Section~\ref{sub:gadgets} is analogous to the NFA derived from the regular expression matching of type $\mid\cdot\mid$ in \cite{BI16}). At presentation level, the difference to earlier work is that we reduce directly from \seth, while the earlier work uses an intermediate problem, orthogonal vectors, as a tool; our reduction can also be presented via the orthogonal vectors problem, but we preferred to work with \seth directly since \sat is more familiar to researcher from various research areas. On a more conceptual level, the new reduction has some interesting features of independent interest. Given a \sat formula, our reduction builds a pattern and a graph, using some special characters in the pattern to match bridges in the graph that can be traversed in one direction only (even if the graph is undirected). Also, obtaining the reduction for a binary alphabet requires a suitable variable-length encoding of the characters to avoid certain paths in the graph.

An earlier version of this reduction can be found in the Master's thesis of the first author \cite{Equi18} (supervised by the two last authors).

\section{Conditional lower bound for PMLG on undirected graphs}
\label{sec:SAT-graphs}

Consider a \sat formula $F$ with variables $v_1, \ldots,  v_n$ and set $C$ of $k$ clauses.\footnote{
In this paper we discuss the interesting case where $k = O(\poly(n))$.
} 
We show how to generate a corresponding instance of \pmlg. We build a pattern $P \in \Sigma^m$ of suitable length $m = \Tilde{O}(2^{\frac n2})$ and a labeled graph $G = (V,E, L)$, where $|E| = \Tilde{O}(2^{\frac n2})$ and $L : V \rightarrow \Sigma^*$ is the node labeling with strings from $\Sigma^*$, such that $P$ matches in $G$ if and only if $F$ is satisfied by some truth assignment of its variables.
Recall that a truth assignment $x$ is a tuple $\langle b_1, \ldots,  b_n \rangle$, where $b_i \in \{\ttrue, \ffalse\}$ is the truth value assigned to each variable $v_i$. We write $x \models c$ to indicate that there exists at least one literal $\ell \in c$ satisfied by $x$ (i.e. either $\ell = v_i$ and $b_i = \ttrue$, or $\ell = \neg v_i$ and $b_i = \ffalse$).

Our reduction builds a pattern with $m = \tilde{O}(2^{\frac n2})$ symbols from a binary alphabet $\Sigma$ along with an undirected graph whose nodes are labeled with single symbols from $\Sigma$ (i.e. $L : V \rightarrow \Sigma$). This graph has $|V|, |E| =\tilde{O}(2^{\frac n2})$ nodes and edges, and maximum degree three. The reduction can be modified so that the graph is directed with maximum sum of indegree and outdegree at least three.

For presentation's sake, we begin with a pattern $P$ using an alphabet of four symbols, $\Sigma = \{\B,\E,\C,\D\}$, whose interpretation is to label nodes according to their implicit functionality: {\B}egin (synchronization token), {\E}nd (synchronization token), {\C}lause (marker), {\D}ummy (don't care); moreover, the resulting undirected graph $G$ has unbounded degree; after that, we will show how to get the minimal degree configuration for $G$ and how to achieve a binary alphabet, as depicted above.

We assume that $n$ is an even number, without loss of generality, and denote by $X$ the set of $2^{\frac n2}$ possible assignments for the first $n/2$ variables, and by $Y$ those for the last $n/2$ variables, that is, 
\begin{align*}
X =& \; \{x_i \, \mid \; x_i = \langle b_1^{(i)}, \ldots,  b_\frac{n}{2}^{(i)}\rangle  \text{\ is a truth assignment for\ } v_1, \ldots,  v_{\frac{n}{2}} \}\text{ and}&  \\  
Y =& \; \{y_j \, \mid \; y_j = \langle b_{\frac{n}{2}+1}^{(j)}, \ldots,  b_n^{(j)}\rangle \text{\ is a truth assignment for\ }  v_{\frac{n}{2}+1}, \ldots,  v_n \}.&
\end{align*}

We call elements of $X$ and $Y$ \emph{half-assignments} and interpret notation $\models$ accordingly. For example, $y_j \models c$ if and only if there is a literal $\ell \in c$ satisfied by the half-assignment $y_j$ (i.e. either $\ell = v_i$ and $b_i^{(j)} = \ttrue$, or $\ell = \neg v_i$ and $b_i^{(j)} = \ffalse$, for some $i \geq \frac n2 +1$).

The reduction components to follow will be interpreted as follows. The pattern encodes by position, placing a symbol \C\ to indicate which clauses \emph{cannot be satisfied} by a half-assignment $x_i \in X$; the other clauses are marked by \D\ as they are already satisfied by $x_i$ alone; symbols \B\ and \E\ are employed to sync the half-assignments from $X$ with portions of the graph, called \emph{gadgets}. 

The gadgets encode which clauses \emph{are satisfied} by the half-assignments of $Y$, encoding each such clause with a distinct node labeled with \C: when a symbol \C\ in the pattern matches a node with label \C\ in the graph, the corresponding clause is now covered by a half-assignment $y_j \in Y$, while it was not yet covered by half-assignment $x_i \in X$.
If all the symbols \C\ for $x_i$ are matched by the nodes of the gadget corresponding to $y_j$, then assignment $x_i y_j$ satisfyes the \sat formula $F$; also the other direction holds. 

Parallel nodes labeled with \D\ are introduced to deal with the cases when the pattern indicates that the corresponding clause is already satisfied by a half-assignment in $X$. Nodes labeled with \B\ or \E\ are used to match a half-assignment $x_i \in X$ with a half-assignment of $y_j \in Y$. Details follow below.

\subsection{Building the pattern}
\label{sub:pattern}

Pattern $P$ is defined over the alphabet $\Sigma = \{ \B,\E,\C,\D \}$ using the half-assignments in $X = \{ x_1, \ldots, x_{2^{\frac n2}}\}$ and the set $C = \{c_1, \dots, c_k\}$ of clauses of \sat formula $F$. Specifically, it is built as the concatenation $P = \E\B P_{x_1}\E\,\B P_{x_2}\E \ldots \B P_{x_{2^{\frac{n}{2}}}}\E\B$ of $2^{\frac n2}$ strings where $x_i \in X$ and, for $1 \leq h \leq k$, the $h$th symbol of string $P_{x_i}$ is defined as
\begin{equation*}
P_{x_i}[h] =
    \begin{cases}
        \C \qquad  \text{if} \quad x_i  \notmodels c_h\\
        \D \qquad  \text{otherwise} 
    \end{cases}
\end{equation*}
We will prove that $F$ is satisfiable if and only if we can find a match for this pattern in our graph, where the latter is made up of gadgets as specified below.

\begin{figure}[t]
\centering
\includegraphics[scale=0.35]{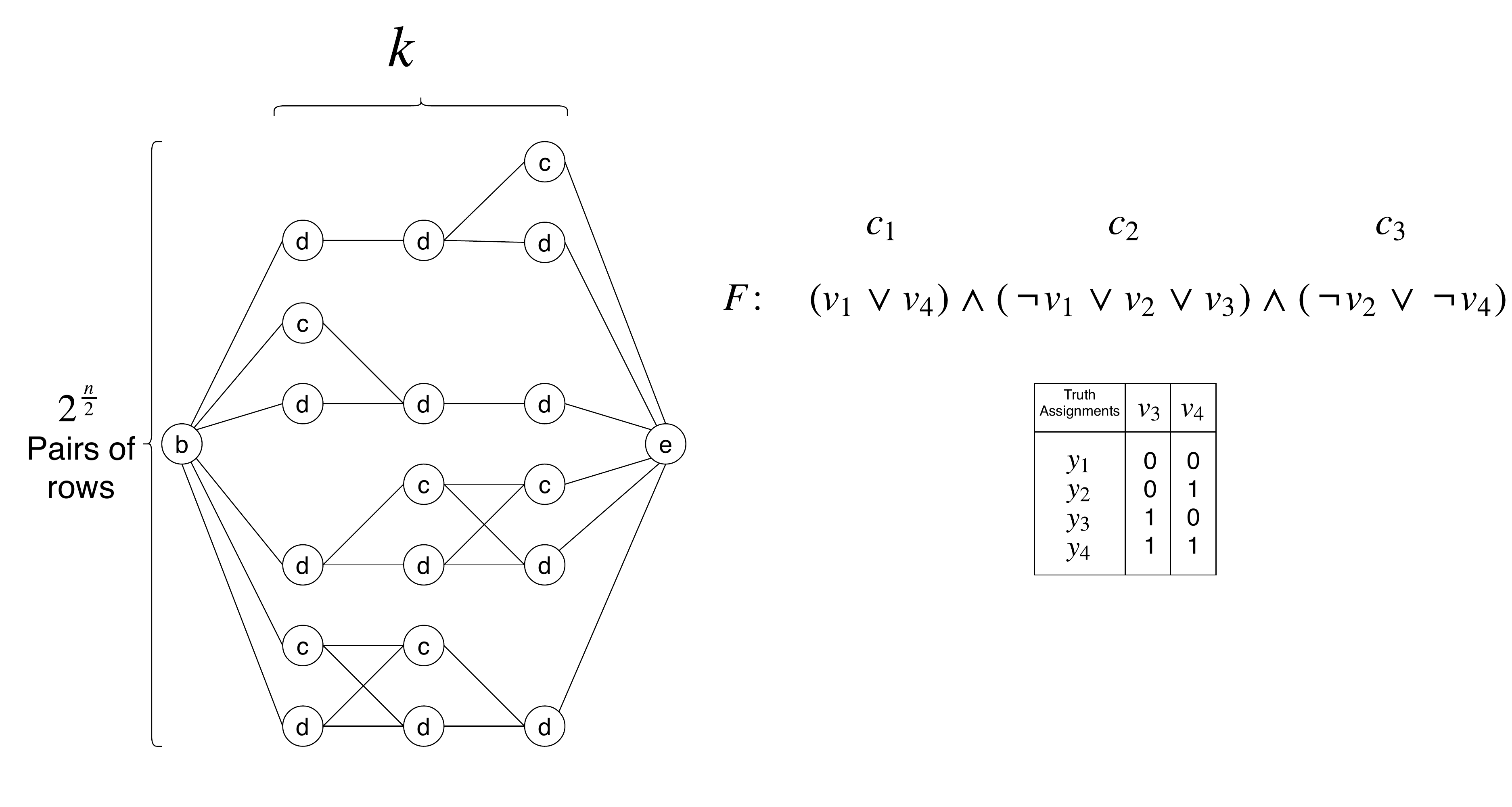}
\caption{Gadget $G_F$ for formula $F$, where $n = 4, 2^\frac n2 = 4$ and $k = 3$. If for instance we look at the first column we observe that $c_{1,1}$ and $c_{3,1}$ are missing meaning that $y_1 \notmodels c_1$ and $y_3 \notmodels c_1$. On the other hand we know that $y_2 \models c_1$ and $y_4 \models c_1$ since we have nodes $c_{2,1}$ and $c_{4,1}$. An example of a pattern that we are able to match is $P = \B \C \C \D \E$ while we would fail on $\bar{P} = \B \C \D \C \E$.}
\label{figure:G_Fgadget}
\end{figure}

\subsection{Graph gadgets for SAT formulas}
\label{sub:gadgets}

Our gadget is an undirected graph $G_F = (V_F,E_F,L_F)$, illustrated in \figurename~\ref{figure:G_Fgadget} and defined as follows using the $2^{\frac n2}$ half-assignments in $Y = \{ y_1, \ldots, y_{2^{\frac n2}}\}$ and the set $C = \{c_1, \dots, c_k\}$ of clauses of \sat formula $F$. 

In the set $V_F$ of nodes, we have a clause node $c_{j,h}$  for every possible pair $y_j,c_h \in Y \times C$ such that  $y_j \models c_h$, and a dummy node $d_{j,h}$ for every possible pair $y_j,c_h \in Y \times C$. Set $V_F$ also contains two special nodes, a begin node $b$ and an end node $e$, 
\[
V_F = \{ c_{j,h} \, \mid \, y_j \models c_h, \; y_j \in Y, c_h \in C \, \} \cup \{ d_{j,h} \, \mid \,  y_j \in Y, c_h \in C \, \} \cup \{ b, e\}
\]

\noindent
Labeling  $L_F: V_F \rightarrow \Sigma$ is consequently defined, where a symbol \C\ in the pattern that matches a node labeled with \C\ in the graph will represent the fact a clause \emph{not} satisfied by a certain half-assignment in $X$ is \emph{actually} satisfied by a certain half-assignment in $Y$. The \D\ symbols are sort of ``don't care'', and \B\ and \E\ symbols synchronize the whole.
\begin{equation*}
L_F(u) = 
    \begin{cases}
        \B \qquad \text{if} \quad u = b\\
        \E \qquad \text{if} \quad u = e\\
        \C \qquad \text{if} \quad u = c_{j,h}\\
        \D \qquad \text{if} \quad u = d_{j,h}\\
    \end{cases}
\end{equation*}

As shown in Fig.~\ref{figure:G_Fgadget}, the edges in the set $E_F$ connect $b$ to every $c_{h,1}$ and $d_{h,1}$, and connect every $c_{h,k}$ and $d_{h,k}$ to $e$, for $1 \leq h \leq k$. Moreover, there is an edge for every pair of nodes that share the same $j$ and are consecutive in terms of $h$ coordinate (e.g. $c_{j,h}, d_{j,h+1}$), for $1 \leq j \leq 2^{\frac n2}$ and $1 \leq h \leq k$. 
\begin{align*}
E_F = &\{ \left(b, c_{j,1}\right) \mid c_{j,1} \in V \} \cup \{ \left(b, d_{j,1}\right) \mid d_{j,1} \in V \} \, \cup 
\{ \left(c_{j,k}, e\right) \mid c_{j,k} \in V \} \cup \{ \left(d_{j,k}, e\right) \mid d_{j,k} \in V \}  \\
&\cup \, \{ \left(c_{j,h}, c_{j,h+1}\right) \mid c_{j,h}, c_{j,h+1} \in V \} \, \cup 
\{ \left(c_{j,h}, d_{j,h+1}\right) \mid c_{j,h}, d_{j,h+1} \in V \} \\
&\cup \, \{ \left(d_{j,h}, c_{j,h+1}\right) \mid d_{j,h}, c_{j,h+1} \in V \} \, \cup 
\{ \left(d_{j,h}, d_{j,h+1}\right) \mid d_{j,h}, d_{j,h+1} \in V\}
\end{align*}

We observe that pattern occurrences in $G_F$ have some combinatorial properties.\footnote{Gadget $G_F$ is analogous to the main component of the \seth reduction to regular expression matching of type $\mid\cdot\mid$ in \cite{BI16}.} 

\begin{lemma} 
\label{lemma:samej}
If subpattern $\B P_{x_i} \E$ matches in $G_F$ then all the nodes matching $P_{x_i}$ share the same $j$ coordinate and have distinct and consecutive $h$ coordinates (i.e. either $c_{j,h}$ or $d_{j,h}$ for $1 \leq h \leq k$).
\end{lemma}
\begin{proof}
Gadget $G_F$ contains a single node $b$ with label $L(b) = \B$ and a single node $e$ with label $L(e) = \E$. Morever, the shortest path from $b$ to $e$ contains $k+2$ nodes ($b$ and $e$ included). As $\B P_{x_i} \E$ contains $k+2$ symbols, its matching path $\pi = b,u_1, \ldots, u_k,e$ in $G_F$ must traverse all distinct nodes by construction. Suppose by contradiction that at least one node in $\pi$ has different $j$ coordinate. This means that two consecutive nodes $u_h$ and $u_{h+1}$ in $\pi$ have coordinates $j$ and $j'$, with $j \neq j'$. Node $u_h$ is actually either $c_{j,h}$ or $d_{j,h}$, whereas 
$u_{h+1}$ is either $c_{j',h+1}$ or $d_{j',h+1}$. By inspection of these four possible cases, we observe that our construction of $G_F$ does not provide any edge connecting $u_h$ and $u_{h+1}$.
Indeed, there is no edge that allows a node to change the $j$ coordinate in the middle of a path. Hence we reach a contradiction. Finally, if one of the matching nodes were not consecutive in terms of $h$ coordinate, by construction we know that we would not be following the shortest path to $e_W^{(j)}$ hence it would not be possible to complete the match.
\end{proof}

\begin{lemma} 
\label{lemma:matchingGF}
Subpattern $\B P_{x_i}\E$ matches in $G_F$ if and only if there is $y_j \in Y$ such that the truth assignment $x_i y_j$ satisfies $F$ (i.e. $x_i y_j \models F$).
\end{lemma}
\begin{proof}
By Lemma~\ref{lemma:samej}, we can focus on the $k$ distinct nodes matching $P_{x_i}$, sharing the same coordinate $j$. We handle the two implications of the statement individually.

($\Rightarrow$) Consider the partial assignment $x_i \in X$. From the structure of the pattern we know that $x_i$ satisfies all the clauses $c_h$ for which $P_{x_i}[h] = \D$. Since $P_{x_i}$ has a match in $G_F$, consider the assignment $y_j \in Y$ where $j$ exists by Lemma~\ref{lemma:samej}, as observed above. We observe that by construction $y_j$ satisfies those clauses that $x_i$ cannot satisfy, namely those for which $P_{x_i}[h] = \C$. Hence we have found a truth assignment $x_i y_j$ that satisfies $F$.
    
($\Leftarrow$) Consider a truth assignment $x_i y_j$ that satisfies $F$, that is, all clauses $c_h$ for \sat formula $F$ are true. Consider now the nodes with coordinate $j$ in $G_F$. For $h = 1, 2, \ldots, k$, if $x_i \models c_h$ then $P_{x_i}[h] = \D$ and matching node $d_{j,h}$ exits in $G_F$ by its definition. If $x_i \notmodels c_h$ then it must be $y_j \models c_h$: thus $P_{x_i}[h] = \C$ and a matching node $c_{j,h}$ exists in $G_F$ by its construction.
The definition of the edges of $G_F$ ensures that all the above nodes $c_{j,h}$ and $d_{j,h}$, as we need, are properly linked to form a path of distinct nodes (for increasing values of $h$); it is so because they all share the same $j$ coordinate. This implies that $P_{x_i}$ matches in $G_F$.  
\end{proof}

\begin{figure}[t]
\centering
\includegraphics[scale=0.40]{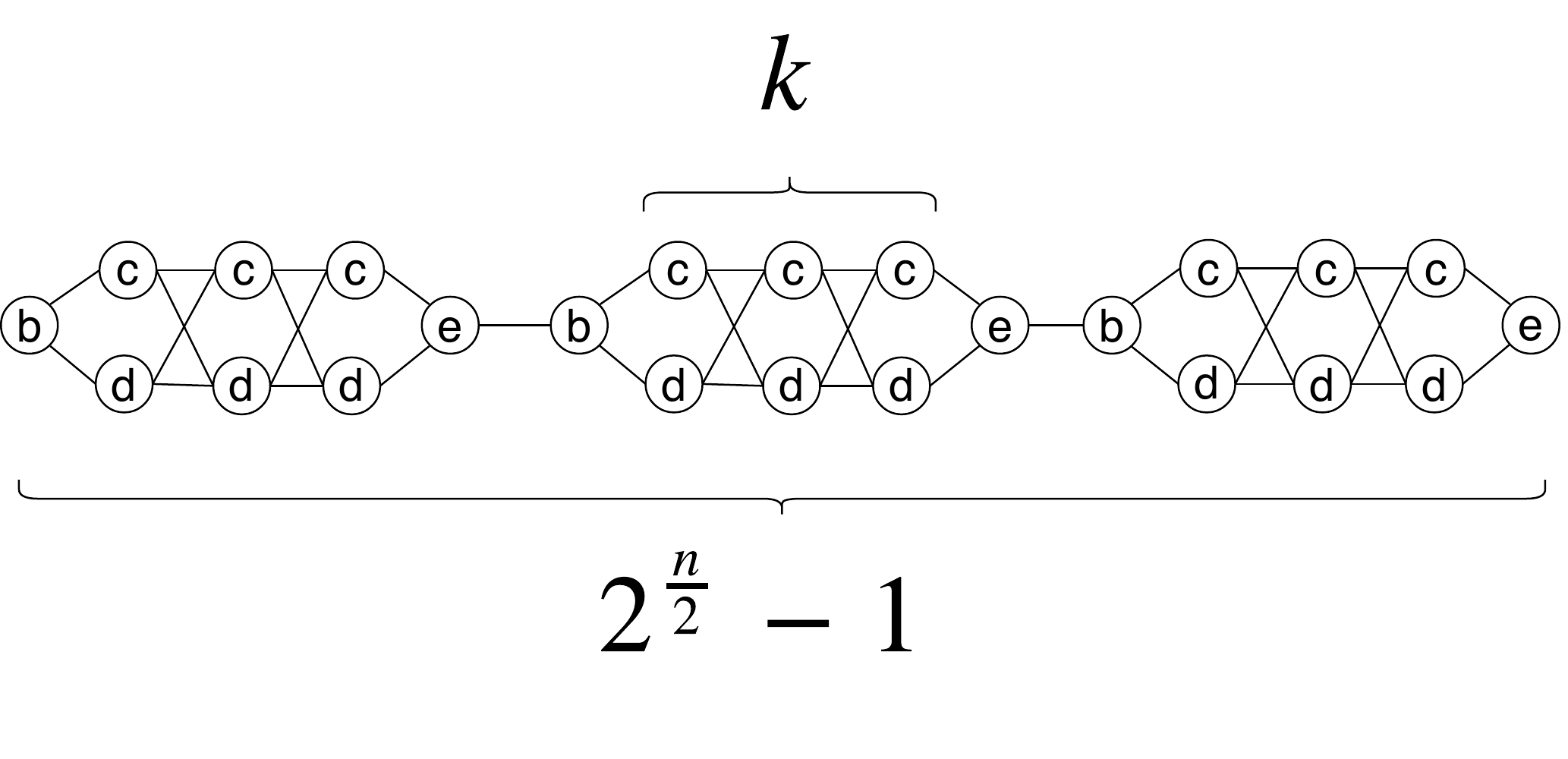}
\caption{In this example $n = 4, 2^\frac n2 = 4$ and $k = 3$. This graph can match any sequence of patterns $P_{x_i}$ but its length is limited to $2^\frac n2 - 1$.}
\label{figure:G_Ugadget}
\end{figure}

While the previous gadget is useful to check whether a half-assignment $x_i$ satisfies $F$ using a \emph{given} subpattern $P_{x_i} \in \B \{ \C, \D \}^k \E$, we need another ``jolly'' gadget that matches \emph{all} subpatterns in $\B \{ \C, \D \}^k \E$ (this is useful when $x_i$ does not satisfy $F$).
We concatenate $2^{\frac{n}{2}}-1$ instances of the latter gadget, thus obtaining the graph $G_U = G(V_U,E_U,L_U)$ illustrated in \figurename~\ref{figure:G_Ugadget}, whose definition is clear from the picture. The $j$th copy of the gadget substructure has a node $b_j$ followed by nodes $c_{j,h}$, $d_{j,h}$ and then node $e_j$, with $1 \leq j \leq 2^{\frac{n}{2}}-1$ and $1 \leq h \leq k$. The labels are $L_U(b_j) = \B$, $L_U(c_{j,h}) = \C$, $L_U(d_{j,h}) = \D$ and $L_U(e_j) = \E$ (we may think about nodes $c_{i,h}$ and $d_{i,h}$ as disposed along two parallel lines). We place the edges $(b_i, c_{i,1}), \, (b_i, d_{i,1}), \, (c_{i,k}, e_j), \, (d_{i,k}, e_i)$ for connecting the beginning and ending nodes of each gadget with its inner part.
We connect nodes $c_{i,h}$ and $d_{i,h}$ with the edges $(c_{i,h}, c_{i,h+1}), \, (c_{i,h}, d_{i,h+1}), \, (d_{i,h}, c_{i,h+1}), \, (d_{i,h}, d_{i,h+1})$. We concatenate our gadgets one after the other using the edges $(e_i, b_{i+1})$, for $i = 1, \ldots, 2^{\frac{n}{2}}-1$.

\begin{figure}[t]
\centering
\includegraphics[scale=0.35]{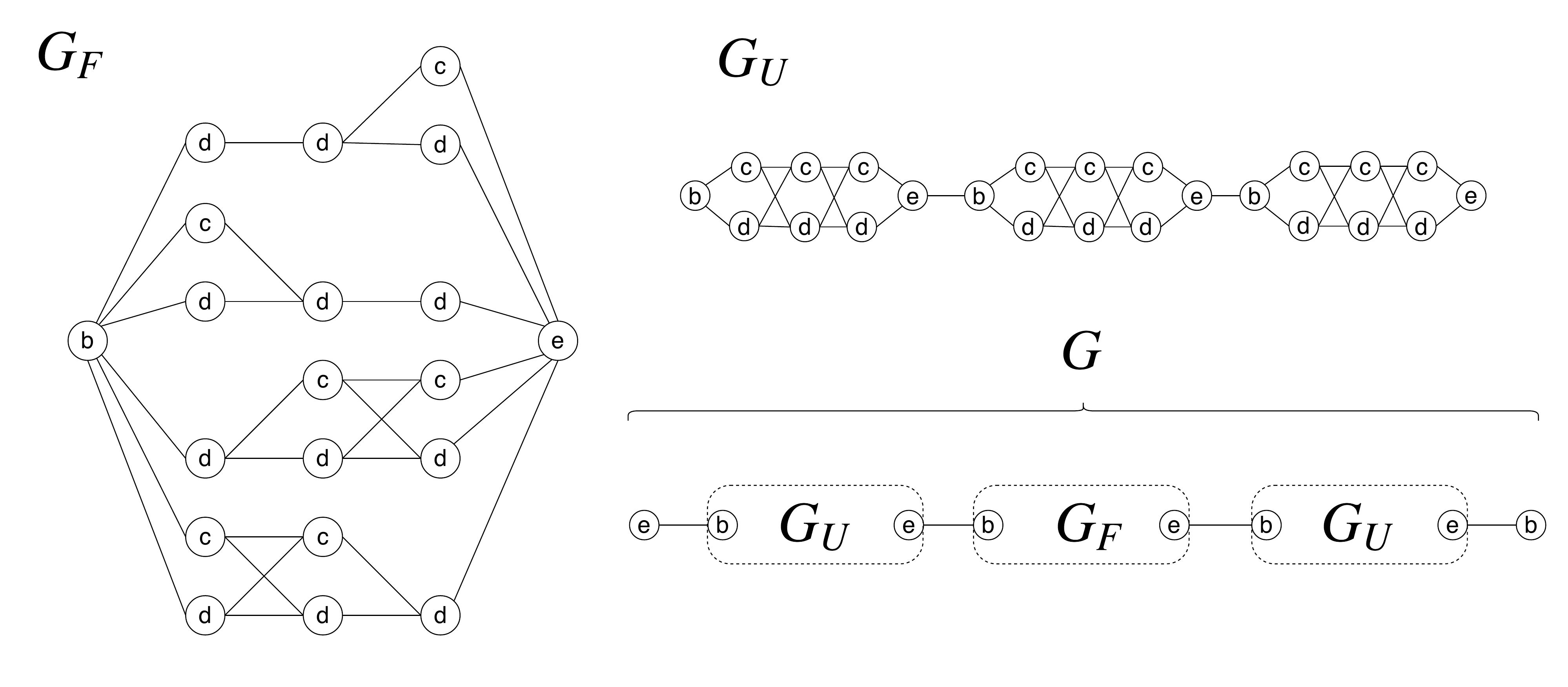}
\caption{This figure shows how to construct $G$ in the example we proposed before, where $n = 4, 2^\frac n2 = 4$ and $k = 3$. What we need to introduce to obtain the final graph $G$ are the two solid edges that are connecting the two instances of $G_U$ with $G_F$ in the bottom part of the figure.}
\label{figure:CompleteGraph}
\end{figure}

\subsection{Putting all together}
\label{sub:all-together}

Armed with gadgets $G_F$ and $G_U$, we obtain the graph $G = (V,E,L)$ from the \sat formula~$F$ by combining them as illustrated in \figurename~\ref{figure:CompleteGraph}. We take one instance of $G_F = (V_F, E_F, L_F)$ and two  instances of $G_U$, say $G_U^{(1)} = (V_U^{(1)}, E_U^{(1)}, L_U^{(1)})$ and $G_U^{(2)} = (V_U^{(2)}, E_U^{(2)}, L_U^{(2)})$, and two new nodes $u$ and $z$, where their label is respectively \E\ and \B. Then $G = (V,E,L)$ has node set $V = V_F \cup \{u,z\} \cup  V_U^{(1)} \cup V_U^{(2)}$, preserving the node labels. The edge set is the union of the previous edge sets plus four edges:  one connects the ``last'' node labeled with \E\ in $G_U^{(1)}$ with the node labeled with \B\ in $G_F$; the other connects the node labeled with \E\ in $G_F$ with the ``first'' node labeled with \B\ in $G_U^{(2)}$, plus $u$ is connected to the first node labeled with \B\ in $G_U^{(1)}$, and the last node labeled with \E\  in $G_U^{(2)}$ is connected to $z$.

\begin{remark}
Each edge connecting a node labeled with \E\ to a node labeled with \B\ is a \emph{bridge} in $G$ (i.e. its removal disconnect $G$). As we shall see, the purpose of these bridges is dual since, within a matching path, the $i$th occurrence of \E\B\ in the pattern matches the $i$th bridge with labels \E\ and \B\ at its endpoints: (i)~they synchronize the distinct subpatterns with the distinct (portions of the) gadgets, and (ii)~they guarantee that the pattern matches a path of distinct nodes rather than a walk.
\end{remark}

We now prove that the reduction is correct, first focusing on subpatterns of $P$.

\begin{lemma} \label{lemma:patternsubpattern}
Pattern $P$ matches in $G$ if and only if a subpattern $\B P_{x_i}\E$ of $P$ matches in $G_F$.
\end{lemma}
\begin{proof}
For the $\Rightarrow$ implication, the bridges with endpoints labeled with \E\ and \B\ can only be traversed once in this direction, as $P$ contains the sequence \E\B\ but does not contain \B\E. Moreover, each occurrence of $P$ must begin with one such bridge and end with another such bridge. For this reason 
each distinct subpattern $\B P_{x_i}\E$ matches a path from either a distinct portion of $G_U^{(l)}$ ($l=1,2$)) or $G_F$. Recall that $G_U^{(1)}$ and $G_U^{(2)}$ can match at most $2^\frac n2 -1$ subpatterns of $P$ each, while $P$ has $2^\frac n2$ of them. Hence one subpattern $\B P_{x_i}\E$ is forced to have a match in $G_F$ in order to have a full match for $P$.

The $\Leftarrow$ implication is trivial. In fact, if $\B P_{x_i}\E$ has a match in $G_F$ then we can match $\B P_{x_1}\E, \ldots, \B P_{x_{i-1}}\E$ in $G_U^{(1)}$  and $\B P_{x_{i+1}}\E, \ldots, \B P_{x_{2^{\frac{n}{2}}}}\E$ in $G_U^{(2)}$ by construction, and have a full match for $P$ in $G$.
\end{proof}

The main result proves the correctness of our reduction.

\begin{theorem} \label{the:correctness}
Pattern $P$ matches in $G$ if and only if the \sat formula $F$ is satisfiable.
\end{theorem}
\begin{proof}
By Lemma~\ref{lemma:patternsubpattern},
$P$ matches in $G$ if and only if a subpattern $P_{x_i}$ matches in $G_F$. By Lemma~\ref{lemma:matchingGF} this holds if and only if the truth assignment $x_iy_j$ satisfies $F$, hence $F$ is satisfiable.
\end{proof}

\subsection{Cost of the reduction}
\label{sub:reduction-cost}

We analyze the cost of building the pattern $P$ and the graph $G$ from the \sat formula $F$.

\begin{lemma}
\label{lemma:cost-reduction}
Given a \sat formula $F$ with $n$ variables, the corresponding pattern $P$ and graph $G$ can be built in $\Tilde{O}(2^\frac n2)$ time and space.
\end{lemma}

\begin{proof}
Checking if an assignment satisfies a clause takes $O(n)$ time which, for our goals, is negligible when compared to $\tilde{O}(2^\frac n2)$. 
Recalling that the number $k$ of clauses is polynomially bounded in $n$, we observe that each $P_{x_i}$ in $P$ has $k$ symbols that can be either \C\ or \D\ plus symbols \B\ and \E. Since $P$ has $2^\frac n2$ sub-patterns $P_{x_i}$, summing everything up we get a length of $m = (k+2) \, 2^\frac n2 = \Tilde{O}(2^\frac n2)$ symbols. As for $G_U$, it has $2^\frac n2$ gadgets each one having $k$ nodes labeled with \C, $k$ nodes labeled with \D, and nodes $b_i$ and $e_i$. Hence there are $(2 + 2k) \, 2^\frac n2 = \Tilde{O}(2^\frac n2)$ total nodes. Each node has a constant number of incident edges (at most $4$) thus their size is $\Tilde{O}(2^\frac n2)$ as well. As for $G_F$, it has $O(k \, 2^\frac n2)$ nodes labeled with \C\ and the same amount of nodes labeled with \D\, plus those with \B\ and \E. In this case, each node has a constant number of edges but for \B\ and \E. Nevertheless, \B\ and \E\ have $O(2^\frac n2)$ edges each, therefore the total amount of edges is again $\Tilde{O}(2^\frac n2)$. For connecting $G_F$ to the two instances of $G_U$ we are adding just $2$ edges. Since the pattern and the graph have size $\Tilde{O}(2^\frac n2)$, we conclude that the cost of our reduction is indeed $\Tilde{O}(2^\frac n2)$.
\end{proof}

\subsection{Implications on SETH}
\label{sub:SETH-implications}

The last step in our proof of Theorem~\ref{theorem:Emlowerbound}
is showing that any $O(|E|^{1-\epsilon} \, m)$-time or $O(|E| \, m^{1-\epsilon})$-time algorithm for \pmlg unavoidably leads to a failure of \seth. To this aim, assume that we have such an algorithm, say $A$. Given a \sat formula $F$ we perform our reduction stated in Theorem~\ref{the:correctness} obtaining pattern $P$ and graph $G$ in $\Tilde{O}(2^\frac n2)$ time by Lemma~\ref{lemma:cost-reduction}, observing that $|E| = \Tilde{O}(2^\frac n2)$ and $m = \Tilde{O}(2^\frac n2)$. At this point, no matter whether $A$ has $O(|E|^{1-\epsilon} m)$ or $O(|E| \, m^{1-\epsilon})$ time complexity, we will end up with an algorithm deciding if $F$ is satisfiable in $\Tilde{O}(2^\frac n2 \, 2^{\frac n2(1-\epsilon)}) = \Tilde{O}(2^{\frac{(2-\epsilon)}{2} n})$ time.
Since $\alpha=\frac{(2-\epsilon)}{2} < 1$ we conclude that this implies to be able to solve \sat in $O(2^{\alpha n})$ time with $\alpha < 1$, making \seth false.

\section{From undirected graphs to DAGs, with binary alphabets}
\label{sec:degree-alphabet-direction}

In this section we show that the graph $G$ obtained from the reduction described in Section~\ref{sec:SAT-graphs} can be transformed so that each node has degree at most three and label chosen from an alphabet of two symbols $\{\zero,\one\}$.

We describe how to modify the proof of Theorem~\ref{theorem:Emlowerbound} so that it holds for any graph of degree at least three. We observe that the graph built in the reduction in Section~\ref{sec:SAT-graphs} has degree $O(2^{\frac n2})$. To obtain degree at most three, we first modify gadgets $G_F$ and $G_U$ to meet such requirement, and then adjust pattern $P$ consequently. Finally, we obtain a binary alphabet for the labels, thus proving Corollary~\ref{cor:general-undirected}.

After that we prove Corollary~\ref{cor:general-dag}, showing that the undirected graph can be easily transformed into a directed acyclic graph (DAG).

\subsection{Maximum degree three}

\paragraph*{Revised gadget $G_F$}
As depicted in Figure \ref{figure:Degree3GFtransf}\textit{(a)}, consider the $O(2^{\frac n2})$ edges connecting node $b$ with nodes $c_{j,1}$ and $d_{j,1}$ in $G_F$. We replace them by a binary tree structure whose nodes are new dummy nodes $f_l$ with labels $L(f_l) = \texttt{d}$ for $1 \le l \le 2^{\frac n2} - 2$. As for node $e$, we proceed along the same way and replace the edges connecting nodes $c_{j,k}$ and $d_{j,k}$ to $e$ by a binary tree structure (this case is not shown in the figure). The internal nodes of these trees have degree at most three.

\begin{figure}[t]
\centering
\includegraphics[scale=0.35]{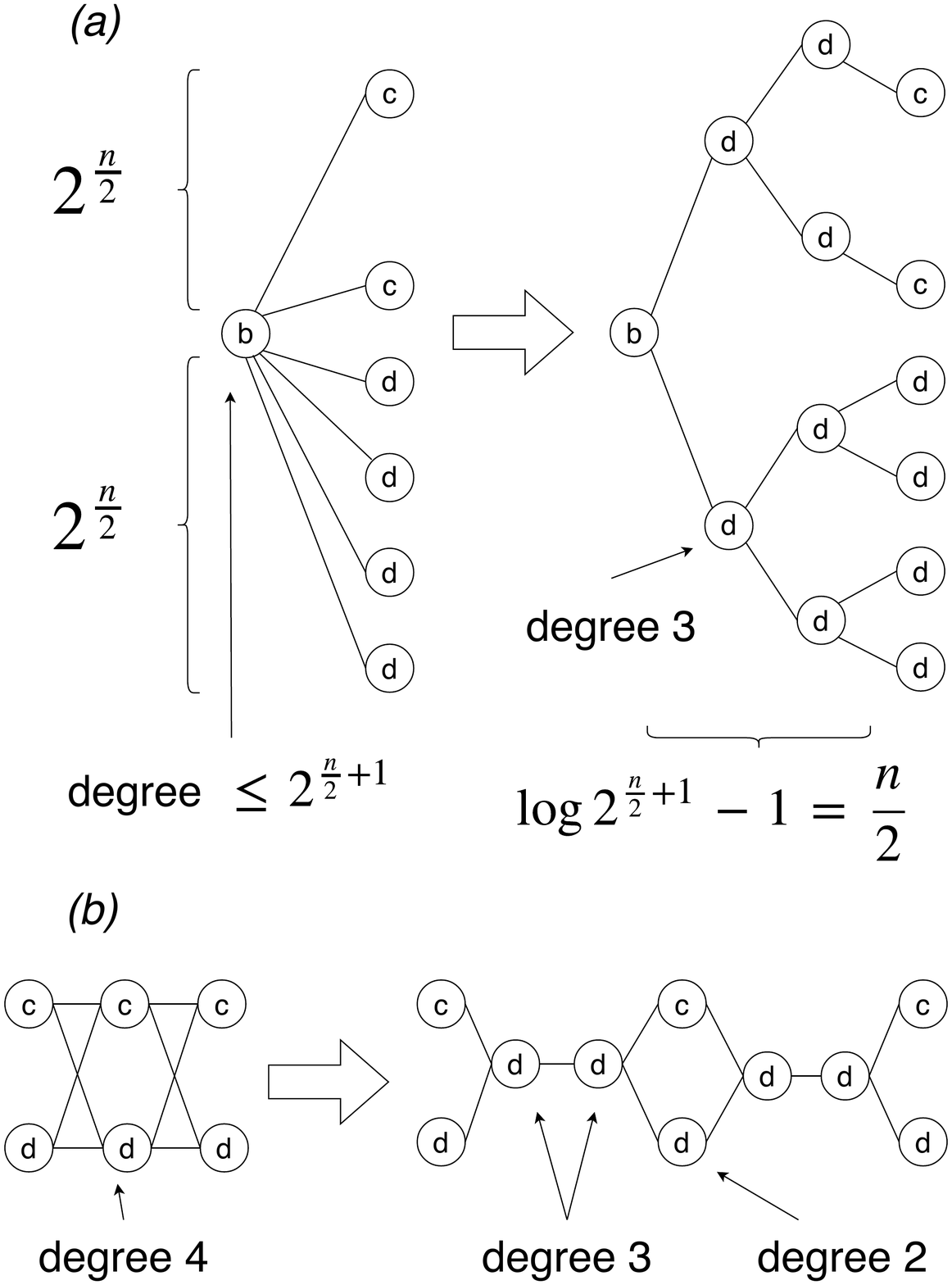}
\caption{The transformation of gadget $G_F$. \textit{(a)} The edges connecting node $b$ with nodes $c_{j,1}$ and $d_{j,1}$ are replaced by a binary tree structure. Note that, even if not reported in this figure, the same process is performed for the edges connecting nodes $c_{j,k}$ and $d_{j,k}$ with node $e$. \textit{(b)} When $c_{j,h-1}, c_{j,h}$ and $c_{j,h+1}$ exist, $c_{j,h}$ and $d_{j,h}$ have degree four. Hence pairs of dummy nodes are added to achieve degree three.}
\label{figure:Degree3GFtransf}
\end{figure}

This is not enough to guarantee degree at most three for each node in $G_F$ as nodes $c_{j,h}$ and $d_{j,h}$ could have degree four. For example, with some nodes $d_{j,h-1}, d_{j,h}$ and $d_{j,h+1}$, nodes $c_{j,h-1}, c_{j,h}$ and $c_{j,h+1}$ could exists. Then both $c_{j,h}$ and $d_{j,h}$ would have degree four. This can be fixed as shown in Figure \ref{figure:Degree3GFtransf}\textit{(b)}, adding two pairs of dummy nodes $f$ with label $L(f)=\D$ to lower the degree to three.\footnote{One pair is placed between nodes $c_{j,h-1}, d_{j,h-1}$ and nodes $c_{j,h}, d_{j,h}$ via edges $(c_{j,h-1},f_{j,h-1}^\textit{(1)}),(d_{j,h-1},f_{j,h-1}^\textit{(1)})$ and $(f_{j,h-1}^\textit{(2)},c_{j,h}),(f_{j,h-1}^\textit{(2)},d_{j,h})$. The other pair of dummy nodes $f$ is placed between nodes $c_{j,h}, d_{j,h}$ and nodes $c_{j,h+1}, d_{j,h+1}$ via edges $(c_{j,h},f_{j,h}^\textit{(1)}),(d_{j,h},f_{j,h}^\textit{(1)})$ and $(f_{j,h}^\textit{(2)},c_{j,h+1}),(f_{j,h}^\textit{(2)},d_{j,h+1})$.}

At this point, we added $O(2^{\frac n2})$ dummy nodes $f$ for the binary tree, and $O((k-1)2^\frac n2) = \Tilde{O}(2^\frac n2)$ pairs of nodes $f_{j,h}^\textit{(1)}, f_{j,h}^\textit{(2)}$. Moreover, the new edges for the binary tree are as many as the nodes while for the other modifications we add one edge for each pair of dummy nodes. The overall time complexity to build the transformed $G_F$ does not increase significantly.

\paragraph*{Revised gadget $G_U$}
Gadget $G_U$ has to be consistent with $G_F$. We add $(\log 2^{\frac n2 + 1}) - 1 = \frac n2$ dummy nodes $f$ with label $L(f)=\D$ between every $b_i$ node and the nodes $c_{i,1}$ and $d_{i,1}$  following it. We also add $\frac n2$ dummy nodes $f$ with label $L(f)=\D$ between every node $e_i$ and the previous nodes $c_{i,k}$ and $d_{i,k}$. We are adding $2\frac n2 (2^\frac n2-1) = \Tilde{O}(2^\frac n2)$ new nodes and one new edge per node, thus the overall time complexity will not be affected. The need for this step will be clearer when we will modify pattern $P$, as it has to match either $G_F$ or $G_U$, so the same format of $P$ is required in both types of gadgets.

We have another issue to handle. As in $G_F$, there could be nodes $c_{i,h}$ and $d_{i,h}$ of degree four. In that case, we add pairs of dummy nodes $f$ with label $L(f)=\D$ following the same schema presented for $G_F$ and illustrated in Figure \ref{figure:Degree3GFtransf}\textit{(b)}. In this way we are introducing $2(k-1)(2^\frac n2 -1) = \Tilde{O}(2^{\frac n2})$ new nodes and one edge for each pair of dummy nodes which do not change the time complexity of the reduction.

\paragraph*{Revised pattern $P$}
Pattern $P = \E\B P_{x_1}\texttt{eb}P_{x_2}\E \ldots \B P_{x_{2^{\frac{n}{2}}}}\E\B$ is modified so as to match $G_F$ and $G_U$ when needed. We add $\frac n2$ symbols $\D$ after each occurrence of $\B$ and before each occurrence of $\E$. Moreover, we insert $\D$ symbols inside the subpatterns $P_{x_i} = a_1 \, a_2 \ldots a_k$, where $a_h \in \{ \texttt{c,d}\}$, to obtain the new subpatterns $P'_{x_i} = a_1 \D \, \D \, a_2 \, \D \, \D \ldots \D \, \D \, a_k$. Therefore, the new pattern to match will be
\begin{equation*}
    P' = \E\B \underbrace{\D \ldots \D}_{{\frac n2} \mathtt{\ times}} P'_{x_1}\underbrace{\D \ldots \D}_{{\frac n2} \mathtt{\ times}} \E \ldots \B \underbrace{\D \ldots \D}_{{\frac n2} \mathtt{\ times}} P'_{x_{2^{\frac{n}{2}}}}\underbrace{\D \ldots \D}_{{\frac n2} \mathtt{\ times}}\E\B.
\end{equation*}

It is worth noting that $P' \in \E\B\, ( \{\C,\D\}^+\E\B)^+$ in this way. The number of new symbols added before and after the subpatterns is $\frac n2 2^\frac n2 = \Tilde{O}(2^\frac n2)$ while the ones inserted inside them are $2(k-1) 2^\frac n2 = \Tilde{O}(2^\frac n2)$. The time cost of the reduction does not increase significantly.

\subsection{Binary alphabet}

The last step consists in defining a binary encoding $\alpha$ of the symbols $\Sigma = \{ \B,\E,\C,\D \}$, namely,
\[
\alpha(\C) = \zero\zero\zero\zero, \quad
\alpha(\D) = \one\one\one\one, \quad
\alpha(\B) = \one\zero, \quad
\alpha(\E) = \zero\one.
\]
Given any string $x = x[1..m]$, we define its binary encoding $\alpha(x) = \alpha(x[1]) \cdots \alpha(x[m])$. The following useful synchronizing property holds, recalling that each edge connecting a node with label \E\ to a node with label \B\ is a bridge in (transformed) $G$. 

\begin{lemma}
\label{lemma:eb-property}
For any string $x \in \Sigma^+$, its binary encoding $\alpha(x)$ contains \zero\one\one\zero\ if and only if $x$ contains \E\B.
\end{lemma}
\begin{proof}
We observe that \E\ and \B\ are encoded by two bits each, while \C\ and \D\ are enconed by four bits each. Hence, \zero\one\one\zero\ can appear by concatenating the binary encoding of two or three symbols. On the other hand, \E\B\ occurs in $x$ if and only if it occurs in a substring of length 3 of $x$. Consequently, it suffices to check the claim by inspection of all the 64 substrings of $x$ of length 3, \C\C\C,~\dots, \E\E\E, and their encodings to see that the property holds. 
\end{proof}

Any walk matched by the revised pattern $P'$ crosses the bridges in the direction from \E\ to \B. 

\begin{lemma}
\label{lem:P-property}
For any pattern $P'$ obtained in the reduction, its binary encoding $\alpha(P')$ does not contain $\one\zero\zero\one = \alpha(\B\E)$.
\end{lemma}
\begin{proof}
Recalling that $P' \in \E\B\, ( \{\C,\D\}^+\E\B)^+$, all the possible substrings of length 3 in $P'$ by construction are of the forms $\B \{\C,\D\}\E$, $\E\B\ \{\C,\D\}$, $\{\C,\D\}\E\B$, $\B \{\C,\D\}^2$, $\{\C,\D\}^2\E$, and $\{\C,\D\}^3$. By inspection of this small number of cases, none contains \B\E, and none of their binary encodings contains \one\zero\zero\one.
\end{proof}

An immediate consequence of Lemma~\ref{lemma:eb-property} and~\ref{lem:P-property} is that the encodings preserve the occurrences. Let $G'$ be the transformed graph, and $P'$ be the revised pattern in the reduction. Let $\alpha(G')$ denote the graph obtained from $G'$ by relabeling its nodes with the binary encoding $\alpha$ of their labels.

\begin{lemma}
\label{lemma:preserving-match}
In the reduction, $P'$ matches in $G'$ if and only if $\alpha(P')$ matches in $\alpha(G')$.
\end{lemma}
\begin{proof}
It follows by Lemma~\ref{lemma:eb-property} and~\ref{lem:P-property}, and the fact that all the edges whose endpoints have one label \E\ and the other label \B\ are bridges, and they are traversed in the direction from \E\ to \B\ when matching $P'$.
\end{proof}

In the encoding above, each node stores two or four bits. By replacing it with a chain of two or four nodes with a single bit as a label, we obtain the proof of Corollary~\ref{cor:general-undirected}.

\subsection{Directed acyclic graphs}
\label{sub:directed}

In order to prove Corollary~\ref{cor:general-dag}, we observe that the proof of Theorem \ref{theorem:Emlowerbound} can be easily modified in order to work also for DAGs.

Considering the definitions of edges $E_F$ and $E_U$ in the proof of Theorem \ref{theorem:Emlowerbound}, and the transformation described so far,  we immediately obtain a directed graph $G'$ that is acyclic. Indeed, because of bridges and occurrences of \E\B\ in the pattern, each pattern match must begin with some bridge, end with a different bridge and lay along a path from the first to the last bridge in the graph. So the edges can be oriented by construction from left (first bridge) to right (last bridge), as it can be checked in Fig.~\ref{figure:G_Fgadget}--\ref{figure:Degree3GFtransf}.

\section{Conclusions\label{sect:conclusions}}

We studied the complexity of pattern matching on labeled graphs, giving a \seth conditional quadratic lower bound for the exact pattern matching. In strings the exact pattern matching takes linear time whereas the approximate pattern matching takes quadratic time under a matching  conditional lower bound. Differently from strings, our result along with the upper bounds in~\cite{AmirLL00,rautiainen2017aligning} imply that the exact and approximate pattern matching (the latter with errors in the pattern) have the same complexity under the \seth conjecture.
Our conditional lower bound uses a binary alphabet and holds even if restricted to nodes of maximum degree at most three for undirected graphs, and to nodes of maximum sum of indegree and outdegree at most three for directed acyclic graphs (DAGs). 

Two border cases are left if the maximum degree or sum of indegree and outdegree is at most two: a) when the undirected graph is a simple path or a cycle, and pattern matching goes along a walk (so it is a sort of zig-zag string matching), and b) when the graph is a directed cycle. For a), we can convert each edge into a pair of arcs and apply the known quadratic algorithm in~\cite{AmirLL00}. On the other hand, we can extend our reduction to derive a matching \seth lower bound for this case \cite{EGMT19}. For b), we can adapt any known string matching algorithm (e.g.~\cite{KMP77}) to get linear time. 

An interesting and natural question for directed graphs is what happens when the graph is deterministic, that is, for each symbol $c$ and each node $v$, there is at most one neighbor of $v$ labeled with $c$. Unfortunately, this does not make the problem any easier. Although our reduction creates an inherently non-deterministic graph, it is possible to alter the reduction scheme to create a deterministic graph \cite{EGMT19}. 

\paragraph*{Acknowledgements}

The first two authors are grateful to Alessio Conte and Luca Versari for providing their comments on the reduction. The last author wishes to thank the participants of the annual research group retreat on sparking the idea to study \seth reductions in this context. We thank the anonymous reviewers of an earlier version of this paper for useful suggestions for improving the presentation and for pointing out the connection to regular expression matching.

\end{document}